%% file: main.tex
\documentclass[runningheads]{llncs}

\usepackage[T1]{fontenc}
\usepackage{amsmath}
\usepackage{amssymb}
\usepackage{graphicx}
\usepackage{hyperref}
\hypersetup{
    colorlinks=true,
    linkcolor=blue,
    urlcolor=blue,
    filecolor=blue,
    citecolor=blue
}
\usepackage{algorithm}
\usepackage{algpseudocode}

\usepackage{tikz}
\usetikzlibrary{shapes.geometric, arrows.meta, positioning, fit, decorations.pathreplacing}
\usepackage{calc}
\usepackage[all]{nowidow}
\usepackage{enumitem}
\usepackage{listings}
\usepackage{tabularx}
\usepackage{multicol}
\usepackage{booktabs}
\usepackage{longtable}
\usepackage[dvipsnames]{xcolor}

\usepackage{microtype}
\microtypesetup{protrusion=true, expansion=true}
\usepackage[none]{hyphenat}


\newcommand{\Pipe}{\mathbf{Pipe}}

\newcommand{\Mult}{\mathbf{Mult}}
\newcommand{\MultH}{\mathbf{Mult}_H}
\newcommand{\PipePair}{\mathbf{PipePair}}
\newcommand{\NB}{\operatorname{NB}}
\newcommand{\Migr}{\operatorname{Migr}}
\newcommand{\inv}{\mathrm{inv}}
\newcommand{\Rpos}{\mathbb{R}_{>0}}
\newcommand{\Rge}{\mathbb{R}_{\ge 1}}
\DeclareMathOperator*{\argmin}{arg\,min}

\title{Constraint Migration: A Formal Theory of Throughput
in AI Cybersecurity Pipelines}
\author{Surasak Phetmanee}
\institute{Department of Electrical and Computer Engineering, Faculty of Engineering, Thammasat School of Engineering, Thammasat University, Thailand\\
\email{psurasak@engr.tu.ac.th}}

\titlerunning{A Formal Theory of Throughput
in AI Cybersecurity Pipelines}
\authorrunning{S. Phetmanee}

\spnewtheorem{assumption}{Assumption}{\bfseries}{\itshape}
\spnewtheorem{microlemma}{Micro Lemma}{\bfseries}{\itshape}

\begin{document}
\maketitle

\begin{abstract}
We develop a formal theory of throughput in finite serial pipeline systems subject to stage multiplicative capacity perturbations, motivated by the deployment of AI tools in cybersecurity operations. A pipeline is a finite totally ordered set of stages each with a positive capacity throughput is the minimum stage capacity.
An admissible multiplier assigns to each stage an improvement factor of at least one.
We prove five theorems and one proposition.
\emph{Theorems~1--2} give exact necessary and sufficient conditions. Throughput is unchanged if and only if at least one bottleneck retains multiplier~$1$, and throughput strictly increases if and only if every bottleneck has multiplier strictly greater than~$1$.
\emph{Theorem~3} establishes that when a nonempty subset of stages is constrained to multiplier~$1$ the human authority constraint, throughput is bounded above by the smallest capacity among those stages, and this bound is tight under unbounded non human acceleration.
\emph{Theorem~4} proves that in a pair of independent attacker defender pipelines, the attacker defender throughput ratio worsens for the defender if and only if the attacker relative throughput gain exceeds the defender.
\emph{Theorem~5} proves that under a fixed false positive fraction model, useful throughput is constant not decreasing above the investigation capacity, establishing that a commonly asserted paradoxical decline is impossible in that model.
\emph{Proposition~6} shows that replacing the fixed fraction with a rate dependent precision function that is strictly decreasing suffices to recover the intended decline.
All proofs are elementary, using only finite minima, real number order properties, and pointwise multiplicative structure.
\end{abstract}

\section{Introduction}
\label{sec:intro}

Informal arguments about AI in cybersecurity---AI speeds up both attackers and defenders, AI cannot replace human judgement, more alerts means more noise---are circulated but not formally grounded.
They are used both to justify and to refute contrary policy positions.
The conditions under which these slogans are true or false remain invisible.

This paper provides formal ground. We model a cybersecurity pipeline as a finite totally ordered set of stages, each with a positive processing capacity.
System throughput is the minimum stage capacity---the deterministic serial bottleneck abstraction familiar from the informal Theory of Constraints~\cite{goldratt1984goal}.
AI is modelled as a stage multiplicative improvement: each stage's capacity is multiplied by a factor of at least one.
Within this framework, we prove exact conditions for when throughput changes, what ceiling human authority stages impose, how attacker and defender throughputs interact, and why a na\"ive false positive model cannot produce the paradoxical decline that informal arguments predict.

\paragraph{Informal statement of main results.}
\emph{Throughput invariance (Theorems~1--2):} Throughput is unchanged if and only if at least one original bottleneck remains unimproved.
Throughput strictly increases if and only if every original bottleneck is strictly improved.
\emph{Human authority ceiling (Theorem~3):} If some stages cannot be accelerated, throughput cannot exceed the smallest capacity among those stages, and this bound is exactly achievable.
\emph{Adversarial acceleration (Theorem~4):} The attacker defender throughput ratio worsens for the defender precisely when the attacker relative gain exceeds the defender.
\emph{False positive impossibility and repair (Theorem~5, Proposition~6):} Under a fixed false positive fraction, useful throughput plateaus rather than declines above saturation; strict decline requires a rate-dependent precision function.

\paragraph{Contributions.}
\begin{itemize}[nosep]
\item A formal pipeline model with explicit axioms, definitions, and notation for throughput, bottleneck sets, multipliers, and perturbed pipelines (\S\ref{sec:prelim}).
\item Exact biconditional characterisations of throughput invariance and strict improvement under multiplicative perturbation (Theorems~\ref{thm:invariance}--\ref{thm:strict}, \S\ref{sec:results}).
\item A tight upper bound theorem for throughput under authority constraints, with an explicit tightness construction (Theorem~\ref{thm:ceiling}, \S\ref{sec:results}).
\item An algebraic equivalence characterising when AI favours the attacker in a two pipeline adversarial model (Theorem~\ref{thm:adversarial}, \S\ref{sec:results}).
\item An impossibility result for the fixed false positive model and a repaired theorem under a rate precision function (Theorem~\ref{thm:impossibility}, Proposition~\ref{prop:repaired}, \S\ref{sec:results}).
\item A scope and limitations analysis identifying what is and is not proved (\S\ref{sec:limits}).
\end{itemize}

\paragraph{Organisation.}
Section~\ref{sec:related} surveys related work.
Section~\ref{sec:prelim} presents the formal model including domains, definitions, and assumptions.
Section~\ref{sec:results} states and proves all results in dependency order.
Section~\ref{sec:discussion} interprets the results.
Section~\ref{sec:limits} discusses limitations and future work.
Section~\ref{sec:conclusion} concludes.

\section{Related Work}
\label{sec:related}

\paragraph{Theory of Constraints.}
The informal Theory of Constraints (TOC) was introduced by Goldratt~\cite{goldratt1984goal} in the context of manufacturing.
Its central claim---that system throughput is determined by the bottleneck, and improving non bottlenecks does not improve throughput---has been adopted in operations management.
However, TOC has not been formalised as a mathematical theory with axioms, definitions, and proved theorems.
This paper provides such a formalisation for the pipeline case, proving exact necessary and sufficient conditions (Theorems~\ref{thm:invariance}--\ref{thm:strict}) rather than informal claims.

\paragraph{Serial production lines.}
The operations research literature on serial production systems~\cite{buzacott1993stochastic} analyses throughput in stochastic settings with buffers and variability.
Our model is deterministic and bufferless including a deliberate simplification that isolates the bottleneck structure from queueing effects.
The throughput formula $T(\Pi) = \min_v c(v)$ is the infinite buffer or fluid limit case of serial line models.
We do not extend to stochastic settings; see Section~\ref{sec:limits}.

\paragraph{Flow networks.}
The max flow min cut theorem~\cite{ford1956maxflow} provides a characterisation of throughput in general networks.
A pipeline is a degenerate case such as the minimum cut equals the minimum stage capacity.
Our contribution is not the throughput formula itself but the perturbation theory built on it characterising invariance, migration, ceilings, and adversarial comparison under multiplicative stage improvements.

\paragraph{Security economics.}
Anderson and Moore~\cite{anderson2001information} and Gordon and Loeb~\cite{gordon2002economics} study security investment through economic models focused on incentives and optimal spending.
Their models do not address pipeline throughput or bottleneck structure.
Our Theorem~\ref{thm:adversarial} provides a throughput ratio criterion for the attacker defender comparison, which complements economic models by identifying when throughput asymmetry favours one side.

\paragraph{Formal methods in security.}
Formal methods have been applied to security protocol verification via the Dolev--Yao model~\cite{dolev1983security}, BAN logic~\cite{burrows1989logic}, and the applied pi calculus~\cite{abadi2001mobile}.
These formalisms address secrecy and authentication properties, not pipeline throughput.
We are not aware of prior formal methods work addressing cybersecurity pipeline bottleneck analysis.

\paragraph{AI in cybersecurity.}
The empirical literature on AI in cybersecurity~\cite{buczak2016survey} documents AI applications in intrusion detection, malware analysis, and threat intelligence.
Alert fatigue and false positive burden are well documented operationally~\cite{alahmadi2022alertfatigue}.
This literature motivates our formal model but does not provide formal theorems.
Our Theorem~\ref{thm:impossibility} and Proposition~\ref{prop:repaired} formalise the false positive burden argument that this literature discusses informally.

\paragraph{Human oversight in security operations.}
The human authority constraint ($\alpha(h)=1$ for human stages) formalises concerns from the human supervisory control literature~\cite{sheridan2002humans,parasuraman1997humans}.
These works argue that certain decision stages cannot be fully automated.
Our Theorem~\ref{thm:ceiling} gives the exact throughput ceiling that follows from this constraint.

\section{Formal Preliminaries}
\label{sec:prelim}

\subsection{Domains}

We write $\Rpos := \{x \in \mathbb{R} : x > 0\}$ for the strictly positive reals and $\Rge := \{x \in \mathbb{R} : x \ge 1\}$ for the admissible multiplicative factors.
Throughout, $V$ denotes a finite nonempty set of \emph{stages}, and $\mathcal{P}(V)$ denotes its power set.

\subsection{Definitions}

\begin{definition}[Pipeline]\label{def:pipe}
A \emph{pipeline} is a tuple $\Pi = (V_\Pi, \prec_\Pi, c_\Pi)$ where
\begin{enumerate}[nosep]
\item $V_\Pi$ is a finite nonempty set (the \emph{stage set}),
\item $\prec_\Pi$ is a strict total order on $V_\Pi$ (the \emph{stage ordering}),
\item $c_\Pi \colon V_\Pi \to \Rpos$ is a function (the \emph{stage capacity function}).
\end{enumerate}
The class of all pipelines is denoted $\Pipe$.
\end{definition}

\begin{remark}
The stage ordering $\prec_\Pi$ models the serial structure of the pipeline.
However, as we note in Section~\ref{sec:discussion}, none of the proofs in this paper invoke $\prec_\Pi$; the theorems depend only on $V_\Pi$ and $c_\Pi$.
The ordering is retained for interpretive fidelity to the pipeline metaphor.
\end{remark}

\begin{definition}[Throughput]\label{def:throughput}
For $\Pi \in \Pipe$, the \emph{throughput} of $\Pi$ is
\[
T(\Pi) := \min_{v \in V_\Pi} c_\Pi(v).
\]
\end{definition}

\begin{definition}[Bottleneck set]\label{def:bottleneck}
For $\Pi \in \Pipe$, the \emph{bottleneck set} is
\[
B(\Pi) := \argmin_{v \in V_\Pi} c_\Pi(v) = \{v \in V_\Pi : c_\Pi(v) = T(\Pi)\}.
\]
The \emph{non bottleneck set} is $\NB(\Pi) := V_\Pi \setminus B(\Pi)$.
\end{definition}

\begin{example}\label{ex:pipeline}
Let $V = \{a,b,c\}$ with $a \prec b \prec c$ and $c_\Pi(a) = 3$, $c_\Pi(b) = 1$, $c_\Pi(c) = 4$.
Then $T(\Pi) = 1$, $B(\Pi) = \{b\}$, and $\NB(\Pi) = \{a,c\}$.
\end{example}

\begin{definition}[Admissible multiplier]\label{def:mult}
For a finite nonempty set $V$, an \emph{admissible multiplier} on $V$ is a function $\alpha \colon V \to \Rge$.
The set of all admissible multipliers on $V$ is denoted $\Mult(V)$.
The \emph{identity multiplier} $\mathbf{1} \in \Mult(V)$ is defined by $\mathbf{1}(v) = 1$ for all $v \in V$.
\end{definition}

\begin{definition}[Perturbed pipeline]\label{def:perturbed}
Let $\Pi \in \Pipe$ and $\alpha \in \Mult(V_\Pi)$.
The \emph{perturbed capacity function} is $c_\Pi^\alpha(v) := \alpha(v) \cdot c_\Pi(v)$ for all $v \in V_\Pi$.
The \emph{perturbed pipeline} is $\Pi^\alpha := (V_\Pi, \prec_\Pi, c_\Pi^\alpha)$.
\end{definition}

\begin{remark}
The superscript $\Pi^\alpha$ denotes perturbation, not exponentiation.
Since $\alpha(v) \ge 1$ and $c_\Pi(v) > 0$, we have $c_\Pi^\alpha(v) > 0$, so $\Pi^\alpha \in \Pipe$.
\end{remark}

\begin{definition}[Constraint migration]\label{def:migration}
Let $\Pi \in \Pipe$ and $\alpha \in \Mult(V_\Pi)$.
\emph{Constraint} occurs under $\alpha$ if $B(\Pi^\alpha) \neq B(\Pi)$.
We write $\Migr(\Pi, \alpha) \iff B(\Pi^\alpha) \neq B(\Pi)$.
\end{definition}

\subsection{Human Authority Extension}

\begin{definition}[Human authority structure]\label{def:human}
Let $\Pi \in \Pipe$.
A \emph{human authority stage set} is any subset $H \subseteq V_\Pi$.
The \emph{human-authority admissible multiplier set} is
\[
\MultH(V_\Pi) := \{\alpha \in \Mult(V_\Pi) : \forall h \in H,\; \alpha(h) = 1\}.
\]
The \emph{human ceiling value} is $C_H(\Pi) := \min_{h \in H} c_\Pi(h)$, defined when $H \neq \varnothing$.
The \emph{machine-stage set} is $M(\Pi, H) := V_\Pi \setminus H$.
\end{definition}

\subsection{Adversarial Extension}

\begin{definition}[Attacker defender pipeline pair]\label{def:adversarial}
An \emph{attacker defender pipeline pair} is $(\Pi_A, \Pi_D) \in \PipePair := \Pipe \times \Pipe$.
We write $V_A := V_{\Pi_A}$, $V_D := V_{\Pi_D}$, $c_A := c_{\Pi_A}$, $c_D := c_{\Pi_D}$.
For $\alpha_A \in \Mult(V_A)$ and $\alpha_D \in \Mult(V_D)$, the \emph{comparative throughput ratio} is
\[
R(\Pi_A, \Pi_D, \alpha_A, \alpha_D) := \frac{T(\Pi_A^{\alpha_A})}{T(\Pi_D^{\alpha_D})},
\]
and the \emph{baseline ratio} is $R_0(\Pi_A, \Pi_D) := T(\Pi_A) / T(\Pi_D)$.
\end{definition}

\subsection{False Positive Model}

This section introduces a scalar model that is formally independent of the pipeline theory above.

\begin{definition}[Simple useful throughput]\label{def:simple-u}
For $\lambda \in \Rpos$, $f \in [0,1)$, and $c_\inv \in \Rpos$, the \emph{simple useful-throughput function} is
\[
U(\lambda, f, c_\inv) := (1 - f) \min(\lambda, c_\inv).
\]
\end{definition}

\begin{definition}[Repaired useful throughput]\label{def:repaired-u}
Let $p \colon \Rpos \to [0,1]$ be a \emph{precision function} and $c_\inv \in \Rpos$.
The \emph{repaired useful-throughput function} is
\[
U_p(\lambda, c_\inv) := p(\lambda) \cdot \min(\lambda, c_\inv).
\]
\end{definition}

\subsection{Assumptions}

We collect the assumptions used across the paper.
Each theorem statement names exactly those assumptions it requires.

\begin{assumption}[Finite nonempty stage set]\label{ax:finite}
For every $\Pi \in \Pipe$, the set $V_\Pi$ is finite and nonempty.
\end{assumption}

\begin{assumption}[Strict positivity]\label{ax:positive}
For every $\Pi \in \Pipe$ and every $v \in V_\Pi$, $c_\Pi(v) > 0$.
\end{assumption}

\begin{assumption}[Throughput is the stagewise minimum]\label{ax:min}
For every $\Pi \in \Pipe$, $T(\Pi) = \min_{v \in V_\Pi} c_\Pi(v)$.
\end{assumption}

\begin{assumption}[Stage multiplicative perturbation]\label{ax:local}
For every $\Pi \in \Pipe$ and $\alpha \in \Mult(V_\Pi)$, $c_\Pi^\alpha(v) = \alpha(v) \cdot c_\Pi(v)$ for all $v \in V_\Pi$.
\end{assumption}

\begin{assumption}[Admissible multipliers are at least one]\label{ax:ge1}
For every $\Pi \in \Pipe$ and $\alpha \in \Mult(V_\Pi)$, $\alpha(v) \ge 1$ for all $v \in V_\Pi$.
\end{assumption}

\begin{assumption}[Human authority stages are unaccelerable]\label{ax:human}
In the human ceiling model, $\alpha(h) = 1$ for all $h \in H$, i.e., $\alpha \in \MultH(V_\Pi)$.
\end{assumption}

\begin{assumption}[Unbounded non human acceleration]\label{ax:unbounded}
For every $M > 1$, there exists $\alpha \in \MultH(V_\Pi)$ with $\alpha(m) \ge M$ for all $m \in V_\Pi \setminus H$.
\end{assumption}

\begin{assumption}[Independent attacker defender pipelines]\label{ax:independent}
The perturbed throughputs $T(\Pi_A^{\alpha_A})$ and $T(\Pi_D^{\alpha_D})$ are computed independently; no coupling exists between the two pipelines.
\end{assumption}

\begin{assumption}[Strictly decreasing precision]\label{ax:monop}
The precision function $p$ is strictly decreasing on $(c_\inv, \infty)$: for all $\lambda_1, \lambda_2$ with $c_\inv < \lambda_1 < \lambda_2$, $p(\lambda_1) > p(\lambda_2)$.
\end{assumption}

\section{Results}
\label{sec:results}

We present all results in the order dictated by the proof dependency graph: foundational lemmas first, then core throughput theorems, then extensions.

\subsection{Lemmas}

\begin{lemma}[Bottleneck existence]\label{lem:botexist}
Let $\Pi \in \Pipe$.
Then $T(\Pi)$ is well defined, $B(\Pi) \neq \varnothing$, and $c_\Pi(v) = T(\Pi)$ for every $v \in B(\Pi)$.
\end{lemma}

\begin{proof}
By Assumption~\ref{ax:finite}, $V_\Pi$ is finite and nonempty.
By Assumption~\ref{ax:positive}, the image $\{c_\Pi(v) : v \in V_\Pi\}$ is a finite nonempty subset of~$\Rpos$.
Every finite nonempty subset of $\mathbb{R}$ has a minimum, so $\min_{v \in V_\Pi} c_\Pi(v)$ exists; by Assumption~\ref{ax:min} this value is $T(\Pi)$.
Since the minimum is attained, there exists $v_0 \in V_\Pi$ with $c_\Pi(v_0) = T(\Pi)$, so $v_0 \in B(\Pi)$ by Definition~\ref{def:bottleneck}, giving $B(\Pi) \neq \varnothing$.
For any $v \in B(\Pi)$, $c_\Pi(v) = T(\Pi)$ holds by the definition of $B(\Pi)$.
\end{proof}

\begin{lemma}[Perturbed pipeline membership]\label{lem:membership}
Let $\Pi \in \Pipe$ and $\alpha \in \Mult(V_\Pi)$.
Then $\Pi^\alpha \in \Pipe$.
\end{lemma}

\begin{proof}
We verify the three conditions of Definition~\ref{def:pipe} for $\Pi^\alpha = (V_\Pi, \prec_\Pi, c_\Pi^\alpha)$.
Conditions~(1) and~(2) hold because $\Pi^\alpha$ inherits $V_\Pi$ and $\prec_\Pi$ from $\Pi \in \Pipe$.
For condition~(3), let $v \in V_\Pi$.
By Assumption~\ref{ax:local}, $c_\Pi^\alpha(v) = \alpha(v) \cdot c_\Pi(v)$.
By Assumption~\ref{ax:ge1}, $\alpha(v) \ge 1 > 0$, and by Assumption~\ref{ax:positive}, $c_\Pi(v) > 0$.
The product of two positive reals is positive, so $c_\Pi^\alpha(v) > 0$.
Therefore $c_\Pi^\alpha \colon V_\Pi \to \Rpos$, and $\Pi^\alpha \in \Pipe$.
\end{proof}

\begin{lemma}[Perturbed throughput normal form]\label{lem:normalform}
Let $\Pi \in \Pipe$ and $\alpha \in \Mult(V_\Pi)$.
Then
\[
T(\Pi^\alpha) = \min_{v \in V_\Pi} \alpha(v) \cdot c_\Pi(v).
\]
\end{lemma}

\begin{proof}
By Lemma~\ref{lem:membership}, $\Pi^\alpha \in \Pipe$.
By Assumption~\ref{ax:min} applied to $\Pi^\alpha$, $T(\Pi^\alpha) = \min_{v \in V_\Pi} c_{\Pi^\alpha}(v)$.
By Definition~\ref{def:perturbed}, $V_{\Pi^\alpha} = V_\Pi$.
By Assumption~\ref{ax:local}, $c_{\Pi^\alpha}(v) = \alpha(v) \cdot c_\Pi(v)$.
Substituting, $T(\Pi^\alpha) = \min_{v \in V_\Pi} \alpha(v) \cdot c_\Pi(v)$.
\end{proof}

\begin{microlemma}[Monotonicity of minimum under domination]\label{ml:mono}
Let $V$ be a finite nonempty set and $f, g \colon V \to \mathbb{R}$ with $f(v) \le g(v)$ for all $v \in V$.
Then $\min_{v \in V} f(v) \le \min_{v \in V} g(v)$.
\end{microlemma}

\begin{proof}
Both minima exist because $V$ is finite and nonempty.
Let $w \in V$ achieve $\min_{v \in V} g(v)$.
Then $\min_{v \in V} f(v) \le f(w) \le g(w) = \min_{v \in V} g(v)$.
\end{proof}

\begin{proposition}[Monotonicity of throughput]\label{prop:mono}
Let $\Pi \in \Pipe$ and $\alpha, \beta \in \Mult(V_\Pi)$ with $\alpha(v) \le \beta(v)$ for all $v \in V_\Pi$.
Then $T(\Pi^\alpha) \le T(\Pi^\beta)$.
\end{proposition}

\begin{proof}
For every $v \in V_\Pi$, since $c_\Pi(v) > 0$ (Assumption~\ref{ax:positive}) and $\alpha(v) \le \beta(v)$, we have $\alpha(v) \cdot c_\Pi(v) \le \beta(v) \cdot c_\Pi(v)$.
By Micro-Lemma~\ref{ml:mono}, $\min_{v} \alpha(v) \cdot c_\Pi(v) \le \min_{v} \beta(v) \cdot c_\Pi(v)$.
By Lemma~\ref{lem:normalform}, the left side is $T(\Pi^\alpha)$ and the right side is $T(\Pi^\beta)$.
\end{proof}

\begin{corollary}[Non decrease under perturbation]\label{cor:nondecrease}
Let $\Pi \in \Pipe$ and $\alpha \in \Mult(V_\Pi)$.
Then $T(\Pi^\alpha) \ge T(\Pi)$.
\end{corollary}

\begin{proof}
By Assumption~\ref{ax:ge1}, $\alpha(v) \ge 1 = \mathbf{1}(v)$ for all $v \in V_\Pi$.
By Proposition~\ref{prop:mono}, $T(\Pi^{\mathbf{1}}) \le T(\Pi^\alpha)$.
By Lemma~\ref{lem:normalform}, $T(\Pi^{\mathbf{1}}) = \min_v 1 \cdot c_\Pi(v) = \min_v c_\Pi(v) = T(\Pi)$.
Therefore $T(\Pi) \le T(\Pi^\alpha)$.
\end{proof}

\subsection{Throughput Theorems}

\begin{microlemma}[Strict lower bound on finite minimum]\label{ml:strict}
Let $V$ be a finite nonempty set, $f \colon V \to \mathbb{R}$, and $c \in \mathbb{R}$.
If $f(v) > c$ for every $v \in V$, then $\min_{v \in V} f(v) > c$.
\end{microlemma}

\begin{proof}
Since $V$ is finite and nonempty, let $w \in V$ achieve $\min_{v \in V} f(v)$.
Then $\min_{v \in V} f(v) = f(w) > c$.
\end{proof}

\begin{theorem}[Throughput invariance characterisation]\label{thm:invariance}
Let $\Pi \in \Pipe$ and $\alpha \in \Mult(V_\Pi)$.
Then
\[
T(\Pi^\alpha) = T(\Pi) \quad\iff\quad \exists\, v \in B(\Pi) \text{ such that } \alpha(v) = 1.
\]
Equivalently, $T(\Pi^\alpha) = T(\Pi)$ if and only if $\min_{v \in B(\Pi)} \alpha(v) = 1$.
\end{theorem}

\begin{proof}
$(\Leftarrow)$\;
Suppose there exists $v_0 \in B(\Pi)$ with $\alpha(v_0) = 1$.
By Lemma~\ref{lem:botexist}, $c_\Pi(v_0) = T(\Pi)$.
By Assumption~\ref{ax:local}, $c_\Pi^\alpha(v_0) = 1 \cdot T(\Pi) = T(\Pi)$.
By Lemma~\ref{lem:normalform}, $T(\Pi^\alpha) = \min_v \alpha(v) \cdot c_\Pi(v) \le \alpha(v_0) \cdot c_\Pi(v_0) = T(\Pi)$.
By Corollary~\ref{cor:nondecrease}, $T(\Pi^\alpha) \ge T(\Pi)$.
Therefore $T(\Pi^\alpha) = T(\Pi)$.
$(\Rightarrow)$\;
We prove the contrapositive: if $\alpha(v) > 1$ for every $v \in B(\Pi)$, then $T(\Pi^\alpha) > T(\Pi)$.
Assume $\alpha(v) > 1$ for all $v \in B(\Pi)$.
By Lemma~\ref{lem:botexist}, $B(\Pi)$ is a finite nonempty subset of $V_\Pi$.
Define $\delta := \min_{v \in B(\Pi)} \alpha(v)$.
By Micro Lemma~\ref{ml:strict} applied with $V = B(\Pi)$, $f = \alpha|_{B(\Pi)}$, and $c = 1$, we have $\delta > 1$.
We show $\alpha(v) \cdot c_\Pi(v) > T(\Pi)$ for every $v \in V_\Pi$ by considering two exhaustive cases.

\emph{Case~1: $v \in B(\Pi)$.}
By Lemma~\ref{lem:botexist}, $c_\Pi(v) = T(\Pi)$.
Since $\alpha(v) \ge \delta$ and $T(\Pi) > 0$ (because $c_\Pi(v) > 0$ for all $v$ by Assumption~\ref{ax:positive}, and $T(\Pi)$ is the minimum of finitely many positive values), we have $\alpha(v) \cdot c_\Pi(v) \ge \delta \cdot T(\Pi) > T(\Pi)$.

\emph{Case~2: $v \in V_\Pi \setminus B(\Pi)$.}
Since $v \notin B(\Pi)$, by Definition~\ref{def:bottleneck}, $c_\Pi(v) \neq T(\Pi)$.
Since $c_\Pi(v) \ge T(\Pi)$ (as $T(\Pi)$ is the minimum) and $c_\Pi(v) \neq T(\Pi)$, we have $c_\Pi(v) > T(\Pi)$.
By Assumption~\ref{ax:ge1}, $\alpha(v) \ge 1$, so $\alpha(v) \cdot c_\Pi(v) \ge c_\Pi(v) > T(\Pi)$.

Since $V_\Pi = B(\Pi) \cup (V_\Pi \setminus B(\Pi))$ and every $v \in V_\Pi$ satisfies $\alpha(v) \cdot c_\Pi(v) > T(\Pi)$, by Micro Lemma~\ref{ml:strict} applied with $V = V_\Pi$ (finite and nonempty by Assumption~\ref{ax:finite}),
\[
T(\Pi^\alpha) = \min_{v \in V_\Pi} \alpha(v) \cdot c_\Pi(v) > T(\Pi).
\]
This completes the contrapositive.

\medskip
\emph{Equivalence of formulations.}
Since $\alpha \in \Mult(V_\Pi)$, we have $\alpha(v) \ge 1$ for all $v \in B(\Pi)$ by Assumption~\ref{ax:ge1}.
Hence $\min_{v \in B(\Pi)} \alpha(v) \ge 1$.
The condition ``$\exists v \in B(\Pi), \alpha(v) = 1$'' holds iff $\min_{v \in B(\Pi)} \alpha(v) \le 1$, which combined with $\ge 1$ gives $\min_{v \in B(\Pi)} \alpha(v) = 1$.
\end{proof}

Throughput stays the same precisely when at least one original bottleneck retains its original capacity, because that stage pins the minimum at its original level while admissibility ($\alpha \ge 1$) prevents any stage from worsening.

\begin{corollary}[Non bottleneck improvement]\label{cor:nb}
Let $\Pi \in \Pipe$ and $\alpha \in \Mult(V_\Pi)$.
If $\alpha(v) = 1$ for all $v \in B(\Pi)$, then $T(\Pi^\alpha) = T(\Pi)$.
\end{corollary}

\begin{proof}
By Lemma~\ref{lem:botexist}, $B(\Pi) \neq \varnothing$, so some $v_0 \in B(\Pi)$ satisfies $\alpha(v_0) = 1$.
By Theorem~\ref{thm:invariance}\;$(\Leftarrow)$, $T(\Pi^\alpha) = T(\Pi)$.
\end{proof}

\begin{theorem}[Strict throughput improvement characterisation]\label{thm:strict}
Let $\Pi \in \Pipe$ and $\alpha \in \Mult(V_\Pi)$.
Then
\[
T(\Pi^\alpha) > T(\Pi) \quad\iff\quad \forall\, v \in B(\Pi),\; \alpha(v) > 1.
\]
\end{theorem}

\begin{proof}
By Corollary~\ref{cor:nondecrease}, $T(\Pi^\alpha) \ge T(\Pi)$, so exactly one of $T(\Pi^\alpha) = T(\Pi)$ or $T(\Pi^\alpha) > T(\Pi)$ holds.
By Theorem~\ref{thm:invariance}, $T(\Pi^\alpha) = T(\Pi) \iff \exists v \in B(\Pi), \alpha(v) = 1$.
Negating: $T(\Pi^\alpha) > T(\Pi) \iff \forall v \in B(\Pi), \alpha(v) \neq 1$.
By Assumption~\ref{ax:ge1}, $\alpha(v) \ge 1$, so $\alpha(v) \neq 1$ is equivalent to $\alpha(v) > 1$.
\end{proof}

Throughput strictly increases if and only if every original bottleneck is strictly improved.
In a system with tied bottlenecks, \emph{all} must be accelerated improving all but one leaves throughput unchanged.

\begin{proposition}[Bottleneck preservation]\label{prop:botpreserve}
Let $\Pi \in \Pipe$ and $\alpha \in \Mult(V_\Pi)$.
Then $B(\Pi^\alpha) = B(\Pi)$ if and only if:
\begin{enumerate}[nosep,label=(\roman*)]
\item $\alpha(u) = \alpha(u')$ for all $u, u' \in B(\Pi)$; and
\item $\alpha(u) \cdot c_\Pi(u) < \alpha(w) \cdot c_\Pi(w)$ for all $u \in B(\Pi)$ and $w \in V_\Pi \setminus B(\Pi)$.
\end{enumerate}
When $B(\Pi) = V_\Pi$, condition~(ii) is vacuously true.
\end{proposition}

\begin{proof}
Let $T := T(\Pi)$ and $T' := T(\Pi^\alpha) = \min_{v} \alpha(v) \cdot c_\Pi(v)$ (Lemma~\ref{lem:normalform}).
$(\Leftarrow)$\;
Let $\alpha_0$ be the common value of $\alpha$ on $B(\Pi)$.
For $u \in B(\Pi)$: $\alpha(u) \cdot c_\Pi(u) = \alpha_0 T$.
For $w \notin B(\Pi)$: $\alpha_0 T < \alpha(w) \cdot c_\Pi(w)$ by~(ii).
Hence $T' = \alpha_0 T$ and the minimisers are exactly $B(\Pi)$, so $B(\Pi^\alpha) = B(\Pi)$.
$(\Rightarrow)$\;
If $B(\Pi^\alpha) = B(\Pi)$, then for every $u \in B(\Pi)$, $\alpha(u) \cdot c_\Pi(u) = T'$.
Since $c_\Pi(u) = T > 0$ for $u \in B(\Pi)$ (Lemma~\ref{lem:botexist}), $\alpha(u) = T'/T$ is the same for all $u$, giving~(i).
For $w \notin B(\Pi)$, since $w \notin B(\Pi^\alpha)$, $\alpha(w) \cdot c_\Pi(w) > T' = \alpha(u) \cdot c_\Pi(u)$, giving~(ii).
\end{proof}

\begin{proposition}[Migration]\label{prop:migration}
Let $\Pi \in \Pipe$ and $\alpha \in \Mult(V_\Pi)$.
Then
\[
\Migr(\Pi, \alpha) \;\iff\; \bigl(\exists v \in B(\Pi) : v \notin B(\Pi^\alpha)\bigr) \;\lor\; \bigl(\exists w \in V_\Pi \setminus B(\Pi) : w \in B(\Pi^\alpha)\bigr).
\]
\end{proposition}

\begin{proof}
By Definition~\ref{def:migration}, $\Migr(\Pi,\alpha) \iff B(\Pi^\alpha) \neq B(\Pi)$.
For any two sets $A, B$: $A \neq B \iff (\exists x \in A \setminus B) \lor (\exists x \in B \setminus A)$.
Applying this with $A = B(\Pi)$ and $B = B(\Pi^\alpha)$ and noting $B(\Pi^\alpha) \subseteq V_\Pi$ gives the result.
\end{proof}

\subsection{Human Authority Ceiling}

\begin{theorem}[Human authority ceiling and tightness]\label{thm:ceiling}
Let $\Pi \in \Pipe$ and $H \subseteq V_\Pi$ with $H \neq \varnothing$.

\textup{(a)}\; For every $\alpha \in \MultH(V_\Pi)$,
\[
T(\Pi^\alpha) \le C_H(\Pi) = \min_{h \in H} c_\Pi(h).
\]

\textup{(b)}\; Under Assumption~\ref{ax:unbounded},
\[
\sup_{\alpha \in \MultH(V_\Pi)} T(\Pi^\alpha) = C_H(\Pi).
\]
\end{theorem}

\begin{proof}
\emph{Part~(a).}\;
By Lemma~\ref{lem:normalform}, $T(\Pi^\alpha) = \min_{v \in V_\Pi} \alpha(v) \cdot c_\Pi(v)$.
Since $H \subseteq V_\Pi$ and $H \neq \varnothing$, $\min_{v \in V_\Pi} \alpha(v) \cdot c_\Pi(v) \le \min_{h \in H} \alpha(h) \cdot c_\Pi(h)$.
By Assumption~\ref{ax:human}, $\alpha(h) = 1$ for $h \in H$, so $\min_{h \in H} \alpha(h) \cdot c_\Pi(h) = \min_{h \in H} c_\Pi(h) = C_H(\Pi)$.
\emph{Part~(b).}\;
\emph{Case~1: $H = V_\Pi$.}\;
Then $\MultH(V_\Pi) = \{\mathbf{1}\}$, so $\sup T(\Pi^\alpha) = T(\Pi) = \min_{v \in V_\Pi} c_\Pi(v) = C_H(\Pi)$.
\emph{Case~2: $M := V_\Pi \setminus H \neq \varnothing$.}\;
Define $c_{\min}^M := \min_{m \in M} c_\Pi(m) > 0$ (exists by Assumption~\ref{ax:finite} and Assumption~\ref{ax:positive}).
Let $N := \lceil C_H(\Pi) / c_{\min}^M \rceil + 1$, so $N \cdot c_{\min}^M > C_H(\Pi)$.
By Assumption~\ref{ax:unbounded}, there exists $\alpha^* \in \MultH(V_\Pi)$ with $\alpha^*(m) \ge N$ for all $m \in M$.
By Lemma~\ref{lem:normalform},
\[
T(\Pi^{\alpha^*}) = \min\!\Bigl(\min_{h \in H} c_\Pi(h),\; \min_{m \in M} \alpha^*(m) \cdot c_\Pi(m)\Bigr) = \min\!\bigl(C_H(\Pi),\; \min_{m \in M} \alpha^*(m) \cdot c_\Pi(m)\bigr).
\]
For every $m \in M$: $\alpha^*(m) \cdot c_\Pi(m) \ge N \cdot c_{\min}^M > C_H(\Pi)$.
By Micro Lemma~\ref{ml:strict}, $\min_{m \in M} \alpha^*(m) \cdot c_\Pi(m) > C_H(\Pi)$.
Therefore $T(\Pi^{\alpha^*}) = C_H(\Pi)$, giving $\sup T(\Pi^\alpha) \ge C_H(\Pi)$.
Combined with Part~(a), $\sup T(\Pi^\alpha) = C_H(\Pi)$.
\end{proof}

\subsection{Adversarial Comparative Statics}

\begin{theorem}[Adversarial relative-acceleration]\label{thm:adversarial}
Let $(\Pi_A, \Pi_D) \in \PipePair$, $\alpha_A \in \Mult(V_A)$, $\alpha_D \in \Mult(V_D)$.
Then
\[
R(\Pi_A, \Pi_D, \alpha_A, \alpha_D) > R_0(\Pi_A, \Pi_D) \quad\iff\quad \frac{T(\Pi_A^{\alpha_A})}{T(\Pi_A)} > \frac{T(\Pi_D^{\alpha_D})}{T(\Pi_D)}.
\]
\end{theorem}

\begin{proof}
All four throughputs $T(\Pi_A)$, $T(\Pi_D)$, $T(\Pi_A^{\alpha_A})$, $T(\Pi_D^{\alpha_D})$ are strictly positive (each is the minimum of finitely many products of positive reals, by Assumptions~\ref{ax:positive} and~\ref{ax:ge1}).
By the definitions of $R$ and $R_0$,
\[
R > R_0 \;\iff\; \frac{T(\Pi_A^{\alpha_A})}{T(\Pi_D^{\alpha_D})} > \frac{T(\Pi_A)}{T(\Pi_D)}.
\]
Multiplying both sides by $T(\Pi_D^{\alpha_D}) / T(\Pi_A) > 0$ gives
\[
\frac{T(\Pi_A^{\alpha_A})}{T(\Pi_A)} > \frac{T(\Pi_D^{\alpha_D})}{T(\Pi_D)}.
\]
\end{proof}

The attacker defender ratio worsens for the defender precisely when the attacker's relative throughput gain exceeds the defender's.

\begin{corollary}[Defender misses bottleneck]\label{cor:defmiss}
If $\alpha_A(v) > 1$ for every $v \in B(\Pi_A)$ and there exists $w \in B(\Pi_D)$ with $\alpha_D(w) = 1$, then $R(\Pi_A, \Pi_D, \alpha_A, \alpha_D) > R_0(\Pi_A, \Pi_D)$.
\end{corollary}

\begin{proof}
By Theorem~\ref{thm:strict}, $T(\Pi_A^{\alpha_A}) > T(\Pi_A)$, so $T(\Pi_A^{\alpha_A})/T(\Pi_A) > 1$.
By Theorem~\ref{thm:invariance}, $T(\Pi_D^{\alpha_D}) = T(\Pi_D)$, so $T(\Pi_D^{\alpha_D})/T(\Pi_D) = 1$.
By Theorem~\ref{thm:adversarial}, $R > R_0$.
\end{proof}

\subsection{False Positive Model: Audit and Repair}

\begin{theorem}[Impossibility of post saturation decline]\label{thm:impossibility}
Let $f \in [0,1)$ and $c_\inv \in \Rpos$.
For every $\lambda_1, \lambda_2 > c_\inv$,
\[
U(\lambda_1, f, c_\inv) = U(\lambda_2, f, c_\inv) = (1-f) \cdot c_\inv.
\]
Consequently, $\lambda \mapsto U(\lambda, f, c_\inv)$ is constant on $(c_\inv, \infty)$.
\end{theorem}

\begin{proof}
Let $\lambda > c_\inv$.
Since $\lambda > c_\inv$, $\min(\lambda, c_\inv) = c_\inv$.
By Definition~\ref{def:simple-u}, $U(\lambda, f, c_\inv) = (1-f) \cdot c_\inv$.
This value does not depend on $\lambda$.
\end{proof}

Under the fixed false positive fraction model, useful throughput saturates at $(1-f) c_\inv$ once the alert rate exceeds investigation capacity.
There is no decline only a plateau.
Any argument asserting strict decline in this model is incorrect.

\begin{proposition}[Repaired false-positive burden]\label{prop:repaired}
Let $p \colon \Rpos \to [0,1]$ satisfy Assumption~\ref{ax:monop}, and let $c_\inv \in \Rpos$.
Then $U_p$ is strictly decreasing on $(c_\inv, \infty)$: for all $\lambda_1, \lambda_2$ with $c_\inv < \lambda_1 < \lambda_2$,
\[
U_p(\lambda_1, c_\inv) > U_p(\lambda_2, c_\inv).
\]
\end{proposition}

\begin{proof}
Let $c_\inv < \lambda_1 < \lambda_2$.
Since both exceed $c_\inv$, $\min(\lambda_i, c_\inv) = c_\inv$ for $i = 1,2$.
By Definition~\ref{def:repaired-u}, $U_p(\lambda_i, c_\inv) = p(\lambda_i) \cdot c_\inv$.
By Assumption~\ref{ax:monop}, $p(\lambda_1) > p(\lambda_2)$.
Since $c_\inv > 0$ (by hypothesis), $p(\lambda_1) \cdot c_\inv > p(\lambda_2) \cdot c_\inv$.
\end{proof}

When precision degrades with increasing alert rate, useful throughput genuinely declines beyond saturation.
The declining precision factor overwhelms the capacity.

\begin{corollary}[Constant precision collapses to plateau]\label{cor:const}
If $p(\lambda) = f_0$ for all $\lambda$ (constant precision), then $U_p(\lambda, c_\inv) = f_0 \cdot c_\inv$ for all $\lambda > c_\inv$.
In particular, $U_p$ is constant, not decreasing, on $(c_\inv, \infty)$.
\end{corollary}

\begin{proof}
For $\lambda > c_\inv$: $U_p(\lambda, c_\inv) = p(\lambda) \cdot c_\inv = f_0 \cdot c_\inv$.
\end{proof}

\section{Discussion}
\label{sec:discussion}

Theorems~\ref{thm:invariance} and~\ref{thm:strict} together provide the complete answer to the question: when does throughput change under multiplicative stage-local improvement?
The answer---at least one bottleneck must be improved for any change, and all must be improved for strict increase---sharpens the informal Theory of Constraints claim into an exact biconditional.
In particular, it handles tied bottlenecks correctly, which informal reasoning typically suppresses.
Theorem~\ref{thm:ceiling} makes precise the common intuition that AI cannot replace human judgement.
The bound $T(\Pi^\alpha) \le \min_{h \in H} c_\Pi(h)$ says that if some stages are designated as requiring human authority, then no amount of machine acceleration can push throughput past the slowest human stage.
The tightness result says this ceiling is not loose: it is exactly achievable.
This has direct implications for security operations staffing such as the human bottleneck is the ultimate constraint.
Theorem~\ref{thm:adversarial} identifies the correct quantity for comparing AI's impact on attackers versus defenders, not raw stage speedup, but \emph{relative throughput gain}.
Corollary~\ref{cor:defmiss} gives the concrete consequence: if the defender invests in non-bottleneck stages while the attacker improves its bottleneck, the defender's investment has zero throughput effect and the ratio worsens.
Theorem~\ref{thm:impossibility} is the most unusual result, which proves that a commonly discussed effect (useful throughput declining with increased detection sensitivity) \emph{cannot occur} under the standard fixed-$f$ model.
The intended behaviour requires a model where precision degrades with rate (Proposition~\ref{prop:repaired}).
We view this as a contribution to modelling methodology: the negative theorem prevents a false claim, and the repair identifies the minimal additional assumption needed.
No proof in this paper invokes the stage ordering $\prec_\Pi$.
The results hold for any finite set of positive capacities under multiplicative perturbation.
We retain the ordering because the motivating application---cybersecurity pipelines---has a natural sequential structure (detection, triage, investigation, response), and because pipeline without ordering would be a misleading metaphor.

\paragraph{Which assumptions are essential.}
The most critical assumption is that throughput equals the stagewise minimum (Assumption~\ref{ax:min}).
Under any other aggregation rule (sum, harmonic mean, max flow), the specific invariance and migration characterisations fail.
The second critical assumption is stage-local multiplicative perturbation (Assumption~\ref{ax:local}): if improving one stage can degrade another through coupling, the invariance theorem is false.
The admissibility constraint $\alpha \ge 1$ (Assumption~\ref{ax:ge1}) is essential for Theorem~\ref{thm:strict} but could be weakened for Theorem~\ref{thm:invariance}; the $(\Leftarrow)$ direction holds even if $\alpha(v) < 1$ is permitted, provided the stage with $\alpha(v_0) = 1$ still pins the minimum.

\section{Limitations and Future Work}
\label{sec:limits}

The theory characterises when throughput changes but not by how much: no closed form expression for $T(\Pi^\alpha) - T(\Pi)$ as a function of the multiplier profile is provided beyond the normal form $T(\Pi^\alpha) = \min_v \alpha(v) c(v)$.
No optimal multiplier allocation under budget constraints is derived.
No stochastic, queueing, or game theoretic extension is attempted.
No empirical validation is provided.
The model is deterministic, bufferless, and serial.
Real cybersecurity pipelines have parallel branches, rework loops, shared resources, and stochastic workloads.
The minimum-throughput axiom captures the bottleneck phenomenon but misses every interaction effect.
The theory should be read as what follows from the bottleneck assumption, not how real SOCs work.
Assumption~\ref{ax:human} ($\alpha(h) = 1$) is an idealisation.
AI may partially assist human stages, yielding $\alpha(h) \in (1, \beta_h]$.
The ceiling theorem generalises cleanly to this setting---$T(\Pi^\alpha) \le \min_{h \in H} \beta_h \cdot c_\Pi(h)$, with the same proof structure but this extension is not developed here.
Theorem~\ref{thm:adversarial} treats attacker and defender as independent.
It does not model strategic interaction, arms-race dynamics, or information asymmetry.
A gametheoretic formulation remains future work.
The false positive section (\S\ref{sec:results}, final subsection) is formally independent of the pipeline theory.
The connection that $c_\inv$ corresponds to a pipeline stage capacity and $\lambda$ to upstream output is interpretive, not proved.
Integrating the two models is a natural extension.

\paragraph{Open problems.}
The most important open problem is extending the theory to stochastic or queueing throughput functionals, where the minimum is replaced by a long run average rate.
Almost none of the current proofs survive in that setting, and new techniques would be required.
A second open problem is the multiplier allocation optimisation: given a cost function on multipliers, which allocation maximises throughput?
This is trivially answered for the deterministic model (invest everything in the bottleneck), but becomes non trivial under stochastic variation, risk aversion, or multi constraints.

\section{Conclusion}
\label{sec:conclusion}

We have presented a formal theory of throughput in finite serial pipelines under multiplicative stage perturbations, proving exact characterisations of throughput invariance (Theorem~\ref{thm:invariance}), strict improvement (Theorem~\ref{thm:strict}), human authority ceilings (Theorem~\ref{thm:ceiling}), adversarial throughput comparison (Theorem~\ref{thm:adversarial}), and false positive model limitations (Theorem~\ref{thm:impossibility}, Proposition~\ref{prop:repaired}).
The contribution is not mathematical depth but formal precision applied to a domain where only informal arguments currently exist.
The theory provides a foundation for evidence based analysis of AI deployment in cybersecurity operations.
The most important open question is whether the core invariance and ceiling results can be extended to stochastic throughput functionals that model the queueing and variability present in real systems.
Such an extension would require fundamentally new proof techniques and constitutes a substantial research programme.


\end{document}